\documentclass[11pt,a4paper]{article}
\usepackage[hidelinks]{hyperref} 
 \usepackage{fullpage}
\usepackage[utf8x]{inputenc}
\usepackage{amsthm}
\usepackage{todonotes,wrapfig,float,graphicx,amssymb,textcomp,array,amsmath}
\usepackage{enumerate,enumitem}
\usepackage{subcaption}
\usepackage{multirow}
\usepackage{lineno}
\usepackage{tabularx}
\usepackage{color,xcolor}

\usepackage{lineno}
\usepackage{cleveref}
\newcommand{\old}[1]{{}}
\newcommand{\later}[1]{{}}

\def\defn#1{\textit{\textbf{\boldmath #1}}}
\renewcommand{\emph}[1]{\defn{#1}} 

\newtheorem{theorem}{Theorem}
\newtheorem{lemma}{Lemma}

\newtheorem{corollary}{Corollary}

\newcommand{\IR}{\mathbb{R}}

\newcommand\height{\ensuremath{\mathrm{height}}}
\newcommand\width{\ensuremath{\mathrm{width}}}

\newcommand{\alg}{\textsf{ALG}}
\newcommand{\opt}{\textsf{OPT}}

\begin{document}
\title{Online Hitting Set for Axis-Aligned Squares}
\author{}
\author{Minati De\thanks{Dept. of Mathematics, Indian Institute of Technology Delhi, New Delhi, India. Email: \texttt{Minati.De@maths.iitd.ac.in}. Research on this paper was supported by   SERB MATRICS Grant MTR/2021/000584.}
\and
Satyam Singh\thanks{Department of Computer Science, Aalto University, Espoo, Finland. Email: \texttt{satyam.singh@aalto.fi}. Research on this paper was supported by the Research Council of Finland, Grant 363444. }
\and
Csaba D. T\'oth\thanks{Department of Mathematics, California State University Northridge, Los Angeles, CA; and Department of Computer Science, Tufts University, Medford, MA, USA. Email: \texttt{csaba.toth@csun.edu}. Research on this paper was supported, in part, by the NSF award DMS-2154347.}
}
\date{}

\maketitle              
\thispagestyle{empty}
\begin{abstract}
We are given a set $P$ of $n$ points in the plane, and a sequence of axis-aligned squares that arrive in an online fashion. The online hitting set problem consists of maintaining, by adding new points if necessary, a set $H\subseteq P$ that contains at least one point in each input square. We present an $O(\log n)$-competitive deterministic 
algorithm for this problem. The competitive ratio is the best possible, apart from constant factors. In fact, this is the first $O(\log n)$-competitive algorithm for the online hitting set problem that works for geometric objects of arbitrary sizes (i.e., arbitrary scaling factors) in the plane. We further generalize this result to positive homothets of a polygon with $k\geq 3$ vertices in the plane and provide an $O(k^2\log n)$-competitive algorithm.
\end{abstract}

\section{Introduction}  \label{sec:intro}

The minimum hitting set problem is one of Karp’s 21 classic NP-hard problems~\cite{GareyJ90}.
In the \emph{minimum hitting set} problem, we are given a set $P$ of elements and a collection $\mathcal{C} = \{S_1, \ldots, S_m\}$ of subsets of $P$, referred to as \emph{ranges}. 
Our aim is to find a set $H\subseteq P$ (\emph{hitting set}) of minimal size such that every set $S_i\in \mathcal{C}$ contains at least one element in $H$. 
Motivated by numerous applications in VLSI design, resource allocation, and wireless networks, researchers have extensively studied the problem for geometric objects.
In the \emph{geometric hitting set} problem, we have $P\subseteq \mathbb{R}^d$ for some constant dimension $d$, and the sets in $\mathcal{C}$ are geometric objects of some type: for example, balls, simplices,  hypercubes, or hyper-rectangles. 
Note that the minimum hitting set problem is dual to the minimum set cover problem in the abstract setting, but duality does not extend to the geometric setting.

In this paper, we study the online hitting set problem for geometric objects.
In the \emph{online geometric hitting set} problem, the point set $P$ is known in advance, while the objects of $\cal C$ arrive one at a time (without advance knowledge). 
We need to maintain a hitting set $H_i\subseteq P$ for the first $i$ objects for all $i\geq 1$. Importantly, in the online setup, points may be added to the hitting set as new objects arrive, 
they cannot be removed (i.e., $H_i\subseteq H_j$ for $i\leq j$). Upon the arrival of a new object $S_{i}\in {\cal C}$, any number of points can be added to the hitting set.  
Depending on whether $P$ is finite~\cite{DeMS24,DeST24,EvenS14,KhanLRSW23} or infinite~\cite{AlefkhaniKM23,CharikarCFM04,DeJKS24,DeS24,DumitrescuGT20,DumitrescuT22}, there are different versions of the online geometric hitting set problem. In this paper, we consider $P$ to be a finite set of points in $\IR^2$.

Let $\alg$ be an algorithm for the online hitting set problem on the instance $(P,\cal C)$. The \emph{competitive ratio} of $\alg$, denoted by $\rho (\alg)$, is the supremum, over all possible input sequences $\sigma$, of the ratio between the size $\alg(\sigma)$ of the hitting set obtained by the online algorithm $\alg$ and the minimum size $\opt(\sigma)$ of a hitting set for the same input\footnote{The \emph{competitive ratio} serves as the \emph{primary measure} of performance of online algorithms, while the \emph{computational complexity} of the algorithm is generally regarded as a \emph{secondary measure}.}: 
\[
    \rho (\alg) = \sup_{\sigma} \left[ \frac{\alg(\sigma)}{\opt(\sigma)}\right].
\]

\subsection{Related Previous Work} 
Alon et al.~\cite{AlonAABN09} initiated the study of the online hitting set problem 
and presented a deterministic algorithm with a competitive ratio of $O(\log |P| \log |\cal C|)$ and obtained almost matching lower bound of $\Omega\left(\frac{\log |P| \log |\cal C|}{\log\log |P| +\log\log |\cal C|}\right)$. While their work addresses the general setting, Even and Smorodinsky~\cite{EvenS14} initiated the study of the online geometric hitting set problem for various geometric objects. They established an optimal competitive ratio of $\Theta(\log |P|)$ when $P$ is a finite subset of $\IR$, and the objects are intervals in $\mathbb{R}$. They also established an optimal competitive ratio of $\Theta(\log |P|)$ when $P$ is a finite subset of $\IR^2$, and the objects are half-planes or congruent disks in the plane. 

Later, Khan et al.~\cite{KhanLRSW23} examined the problem for a finite set of integer points $P \subseteq [0,N)^2 \cap \mathbb{Z}^2$ and a collection $\cal C$ of axis-aligned squares $S \subseteq [0, N)^2$ with integer coordinates for $N>0$. 
They developed an $O(\log N)$-competitive deterministic algorithm for this variant.
They also established a randomized lower bound of $\Omega(\log |P|)$, where $P\subset\mathbb{R}^2$ is finite and $\cal C$ consists of translates of an axis-aligned square. 
De et al.~\cite{DeMS24,DeST24} further investigated the problem for a finite set $P\subset \mathbb{R}^2$, where the collection $\mathcal{C}$ consists of geometric objects with scaling factors (e.g., diameters) in the interval $[1, M]$ for some parameter $M>0$. In~\cite{DeMS24}, they considered homothetic copies of a regular $k$-gon (for $k \geq 4$) and developed a randomized algorithm with expected competitive ratio $O(k^2 \log M \log |P|)$. Although regular $k$-gons can approximate disks as $k \to \infty$, this result does not imply a competitive algorithm for disks with radii in $[1, M]$. In~\cite{DeST24}, they addressed this gap by presenting an $O(\log M \log |P|)$-competitive deterministic algorithm for homothetic disks, and further generalized their result to positive homothets of any convex body in the plane with scaling factors in $[1, M]$.


\begin{table}[htb]
    \centering
    \begin{tabular}{||p{2.7cm}|p{4 cm}|p{2.3 cm}|p{3.3 cm}||} 
 \hline
 Finite Point Set & Objects & Lower Bound & Upper Bound \\ [0.5ex] 
 \hline\hline
 $P\subset \IR$ & Intervals in $\IR$ & $\Omega(\log |P|)$~\cite{EvenS14} & $O(\log |P|)$~\cite{EvenS14} \\ 
 \hline
 $P\subset \IR^2$ & Half-planes in $\IR^2$ & $\Omega(\log |P|)$~\cite{EvenS14} & $O(\log |P|)$~\cite{EvenS14} \\ 
 \hline
 $P\subset \IR^2$ & Congruent disks in $\IR^2$ & $\Omega(\log |P|)$~\cite{EvenS14} & $O(\log |P|)$~\cite{EvenS14} \\ 
 \hline
 $P\subseteq [0, N)^2\cap\mathbb{Z}^2$ & Axis-aligned squares 
 with integral vertices &  $\Omega(\log |P|)$~\cite{KhanLRSW23} $(\#)$ & $O(\log N)$~\cite{KhanLRSW23}          \\
 \hline
 $P\subseteq [0, N)^2\cap\mathbb{Z}^2$ & Bottomless rectangles \,\,\,\,\,\,\,\,\,\,
 (of the form $[a,b]\times [-\infty,c]$)
 & $\Omega(\log |P|)$~\cite{EvenS14} & $O(\log N)$~\cite{DeST24}         \\
 \hline
  $P\subset \IR^2$ & Positive homothets of an arbitrary convex body in $\IR^2$ with scaling factors in the interval $[1, M]$ & $\Omega(\log |P|)$~\cite{KhanLRSW23} $(\#)$ & $O(\log M\log |P|)$~\cite{DeST24} \\ 

 \hline
 \hline \hline
 $P\subset \IR^2$ & Axis-aligned squares
& $\Omega(\log |P|)$~\cite{KhanLRSW23} $(\#)$ & $O(\log |P|)$ \hspace{1cm} [\Cref{thm:main}] \\ 
  \hline
$P\subset \IR^2$ & Axis-aligned rectangles of aspect ratio at most  $\varrho\geq 1$ 
& $\Omega(\log |P|)$~\cite{KhanLRSW23} $(\#)$ & $O(\varrho\log |P|)$ \hspace{1cm} [\Cref{thm:main}] \\ 
  \hline
$P\subset \IR^2$ & Positive homothets of a polygon with $k\geq 3$ vertices
& $\Omega(\log |P|)$~\cite{KhanLRSW23} $(\#)$ & $O(k^2\log |P|)$ \hspace{1cm} [\Cref{thm:generalization}] \\ 
 \hline
\end{tabular}
    \caption{Summary of known and new results for the geometric online hitting set problem where $|P|=n$ for some $n\in \mathbb{N}$. $(\#)$ indicates lower bounds for randomized algorithms. Our results are listed in the last three rows.}
    \label{table_1}
    \vspace{-.5 cm}
\end{table}

\subsection{Our Contribution}
We present the first $O(\log n)$-competitive algorithm for the online hitting set problem for a set of $n$ points and geometric objects of arbitrary sizes in the plane; \Cref{table_1} summarizes previous and new results. 
Our algorithm works for axis-aligned squares of arbitrary sizes, and generalizes to axis-aligned rectangles of bounded aspect ratio. 
The \emph{aspect ratio} of a rectangle is the ratio of the length of the longer to that of the shorter side (e.g., the aspect ratio of a square is 1, and the aspect ratio of a $1\times 2$ or a $2\times 1$ rectangle is 2).

\begin{theorem}\label{thm:main}
For every $\varrho\geq 1$, there is an $O(\varrho\log n)$-competitive deterministic algorithm for the online hitting set problem for any set of $n$ points in the plane and a sequence of axis-aligned rectangles of aspect ratio at most $\varrho$. 
\end{theorem}

We further generalize \Cref{thm:main} to positive homothets of a polygon. 
\begin{theorem}\label{thm:generalization}
        Let $M$ be a polygon with $k\geq 3$ vertices. Then there is an $O(k^2\log n)$-competitive deterministic algorithm for the online hitting set problem for any set of $n$ points in the plane, and a sequence of positive homothets of $M$.  
\end{theorem}

The previous best competitive ratio for these problems was $O(\log^2 n)$, by Alon et al.~\cite{AlonAABN09}, which holds more generally for any collection of sets of polynomial size, including any collection of geometric objects of bounded VC-dimension.
Even and Smorodinsky~\cite{EvenS14} asked whether there is an online hitting set algorithm with $O(\log n)$ competitive ratio for any set system on $n$ points with bounded VC-dimension.
Khan et al.~\cite{KhanLRSW23} asked whether there is an online algorithm for hitting set for squares with a competitive ratio of $O(\log n)$, 
as their algorithm is restricted to integer points in $[0, N)^2\cap \mathbb{Z}^2$ and
 it is $O(\log N)$-competitive even if $|P|\ll N$.
Our result (\Cref{thm:main}) gives an affirmative answer to their question.
It is unclear whether $O(\log n)$-competitive online algorithms are possible for any other families of geometric set systems. We briefly discuss roadblocks to possible generalizations to disks (of arbitrary radii) in 2D or to cubes in 3D in~\Cref{sec:conclusion}. 

\smallskip\noindent\emph{Technical highlights.} The key technical tool for an $O(\log n)$-competitive online algorithm for squares (and axis-aligned rectangles of bounded aspect ratio) is the classical BBD tree data structure by Arya et al.~\cite{AryaM00}, which is computed for the given set $P$ of $n$ points in the plane. It is a hierarchical space partition of depth $O(\log n)$, where all sets are ``fat'' (in the sense defined below). When axis-aligned rectangles arrive one by one in an online fashion, we choose hitting points in a top-down traversal of the BBD tree.  For each point $p\in \opt$ of an (unknown) minimum hitting set $\opt$, the BBD tree is ``saturated'' after $O(\log n)$ levels: we prove that our algorithm uses only $O(\log n)$ points for any sequence of rectangles that can be optimally hit by a single point. See \Cref{sec:alg} for further details.

\smallskip
\noindent\emph{Organization.} \Cref{sec:pre}, begins by introducing necessary definitions and then reviews a classical space partition data structure, the Balanced Box Decomposition Tree (BBD tree)~\cite{AryaMNSW98}. \Cref{sect_prop} presents several key properties of BBD trees. \Cref{sec:alg} describes our online algorithm for hitting axis-aligned rectangles, and analyzes its competitive ratio. 
Then, \Cref{sec:generalizations}, generalizes the main result from axis-aligned squares to positive homothets of an arbitrary polygon.  Finally, \Cref{sec:conclusion} concludes with a discussion of future research directions.

\section{Notation and Preliminaries}
\label{sec:pre}

Unless stated otherwise, the term \emph{object} refers to a compact set in $\IR^d$ with a nonempty interior. Let $\sigma$ denote such an object.
For a scaling parameter $\lambda\in \mathbb{R}$ and a translation vector $b\in\IR^d$, the set $\lambda \sigma+b=\{\lambda x+b : x\in \sigma\}$ is called a \emph{homothet} or \emph{homothetic copy} of $\sigma$; and it is a \emph{positive homothet} if $\lambda>0$. 

\medskip\noindent
\emph{BBD Trees.} 
Arya et al.~\cite{AryaMNSW98} introduced the \emph{Balanced Box Decomposition Tree} (\emph{BBD tree}, for short), which is a binary space partition tree for a set of $n$ points in $\mathbb{R}^d$. Since its introduction in the 1990s, BBD trees have become a widely used data structure for processing and classifying spatial data in computational geometry and beyond. 
In contrast to the quadtree (or compressed quadtree), the depth of the BBD tree is $O(\log n)$ and the nodes correspond to ``fat'' regions; the precise definition is below. 

\begin{figure}[ht]
    \centering
    \includegraphics[scale=1]{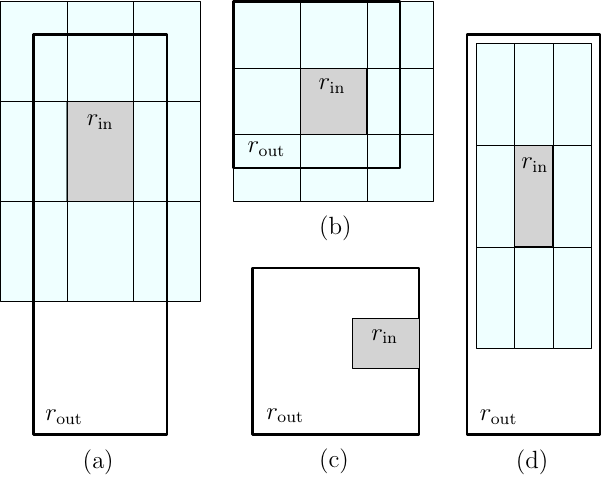}
    \caption{The rectangle $r_{\rm in}$ is not sticky for the rectangle $r_{\rm out}$ in (a) and (b), and sticky in (c) and (d).}
    \label{fig_sticky}
    \vspace{-.35 cm}
\end{figure}

For a set $P$ of $n$ points in an axis-aligned square (bounding box), the BBD tree is a binary tree $T$, where the nodes correspond to regions, called \emph{cells} (the root corresponds to the bounding box). The parent-child relation corresponds to containment between the corresponding cells with the following properties:
\begin{itemize}
     \item Each node $v\in V(T)$ corresponds to a \emph{cell} $C_v=r_{\rm out}(v)\setminus r_{\rm in}(v)$, where $r_{\rm in}(v)$ and $r_{\rm out}(v)$ are axis-aligned rectangles such that $r_{\rm in}(v)\subset r_{\rm out}(v)$, and possibly $r_{\rm in}(v)=\emptyset$.
    \item the aspect ratio of $r_{\rm out}(v)$ and $r_{\rm in}(v)$ (if $r_{\rm in}\neq \emptyset$) is at most 3;
    \item if $r_{\rm in}(v)\neq \emptyset$, then it is \emph{sticky}, which means that the vertical (respectively, horizontal) distance between $r_{\rm in}(v)$ and the boundary of $r_{\rm out}(v)$ is either 0 or at least the side length of $r_{\rm in}(v)$. An equivalent condition for stickiness can be obtained by considering the regular grid consisting of $3^2$ translated copies of $r_{\rm in}(v)$, centered around $r_{\rm in}(v)$. The rectangle $r_{\rm in}(v)$ is sticky for $r_{\rm out}(v)$ if and only if every copy of $r_{\rm in}(v)$ in this grid either lies entirely within $r_{\rm out}(v)$ or is disjoint from the interior of $r_{\rm out}(v)$, see Figure~\ref{fig_sticky};
    \item the cells $\{C_v: v\in V(T)\}$ form a laminar set system, that is, if $u$ is a descendant of $v$, then $C_u\subset C_v$, otherwise ${\rm int}(C_u)\cap {\rm int}(C_v)=\emptyset$;
    \item each leaf node $v\in V(T)$, the region $C_v$ contains at most one point in $P$.
\end{itemize} 
Each internal node has exactly two children, generated by one of two operations: a \emph{fair split}, which decomposes a cell $C_v$ along an axis-parallel line into two cells, or a \emph{shrink}, which introduces a new box $R$ such that $r_{\rm in}(v)\subseteq R\subseteq r_{\rm out}(v)$ and decomposes $C_v=r_{\rm out}(v)\setminus r_{\rm in}(v)$ into $r_{\rm out}(v)\setminus R$ and $R\setminus r_{\rm in}(v)$. 
Furthermore, the number of nodes in $T$ is $O(n)$, the depth of $T$ is $O(\log n)$, and the entire structure can be constructed in $O(n\log n)$ time~\cite{AryaMNSW98}.



\section{A Rectangle amid Cells of the BBD Tree}
\label{sect_prop}
In this section, we present some important properties of BBD trees and discuss the relative position of an arbitrary axis-aligned rectangle of bounded aspect ratio with respect to the cells of the BBD tree. These properties play a crucial role in the design and analysis of our algorithms in \Cref{sec:alg,sec:generalizations}.

\subsection{Crossing Between Rectangles and Cells of the BBD Tree}
Let $T$ be a BBD tree for a finite point set $P$. We say that an axis-aligned rectangle $R$ \emph{crosses} a cell $C_v$, $v\in V(T)$, if $C_v\cap R\neq \emptyset$ but $C_v$ does not contain any vertex of $R$ and $R$ does not contain any vertex of $C_v$ (i.e., vertices of $r_{\rm out}(v)$ and vertices of $r_{\rm in}(v)$ if any). See~\Cref{fig_crossing} for examples.

\begin{figure}[htbp]
    \centering
    \includegraphics[page=2,scale=1]{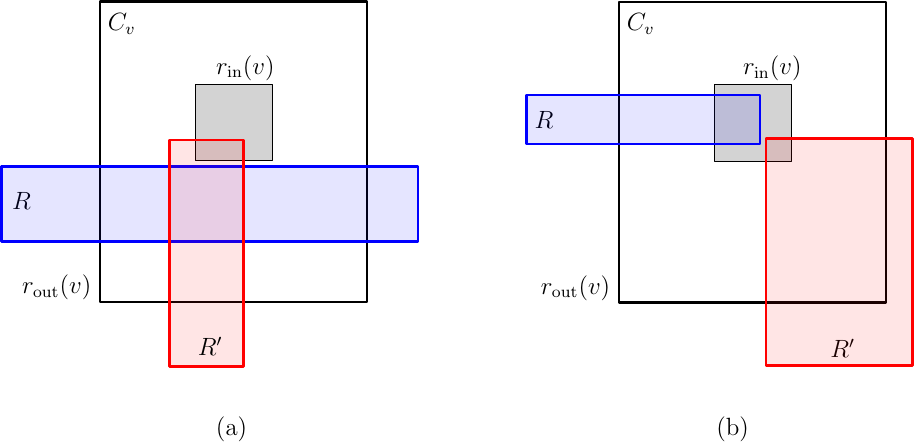}
    \caption{In both (a) and (b), the rectangle $R$ crosses the cell $C_v$, but $R'$ does not.}
    \label{fig_crossing}
    \vspace{-.35 cm}
\end{figure}

\begin{lemma}\label{lem:rho}
    If $R$ is an axis-aligned rectangle of aspect ratio at most $\varrho$, for some $\varrho\geq 1$, then it crosses $O(\varrho)$ interior-disjoint cells of a BBD tree. 
\end{lemma}
\begin{proof}
Let $T$ be a BBD tree, and let $U\subset V(T)$ be a set of nodes that corresponds to a family of interior-disjoint cells crossed by $R$. Let $v\in U$.
Recall that a cell of the BBD tree is defined as $C_v=r_{\rm out}(v)\setminus r_{\rm in}(v)$, where $r_{\rm in}(v)$ may be empty. Depending on how $R$ intersects with $C_v$, we distinguish between two cases.

\smallskip
\noindent \textbf{Case~1: $R$ crosses $C_v=r_{\rm out}(v)\setminus r_{\rm in}(v)$ and $R$ intersects two opposite edges of $r_{\rm out}(v)$.} Assume w.l.o.g.\ that $R$ intersects two vertical edges of $r_{\rm out}(v)$ (i.e., $R$ crosses $r_{\rm out}(v)$ vertically); see~\Cref{fig_crossing}(a). Then we have 
\[
    \width(R)\leq \varrho \cdot \height(R)
    <\varrho\cdot \height(r_{\rm out}(v)) 
    \leq 3\varrho\cdot \width(r_{\rm out}(v)).
\]

In particular, this implies $\width(R)< 3\varrho\cdot \width(r_{\rm out}(v))$. Note also that $\width(R\cap r_{\rm out}(v))=\width(r_{\rm out}(v))$. By the pigeonhole principle, $R$ crosses horizontally at most $3\varrho$ interior-disjoint cells. Similarly, $R$ may cross at most $3\varrho$ cells vertically. However, if $R$ crosses a cell $C_v$ horizontally and another cell $C_w$ vertically, then ${\rm int}(C_v)\cap {\rm int}(C_w)\neq \emptyset$. In particular, $U$ cannot contain both $v$ and $w$. Overall, $R$ crosses at most $3\varrho$ interior-disjoint cells in this case.

\smallskip\noindent \textbf{Case~2: $R$ crosses $C_v=r_{\rm out}(v)\setminus r_{\rm in}(v)$ and $R$ intersects a pair of parallel edges in $r_{\rm out}(v)$ and $r_{\rm in}(v)$, respectively.} 
Assume w.l.o.g.\ that $R$ intersects the left side of both $r_{\rm in}(v)$ and $r_{\rm out}(v)$; see~\Cref{fig_crossing}(b). Let $\mathrm{dist}_{\rm left}(v)$ denote the distance between the left sides of $r_{\rm in}$ and $r_{\rm out}$. Then 
we have $\width(R\cap C_v)={\rm dist}_{\rm left}(v)$, and in particular ${\rm dist}_{\rm left}(v)<\width(R)$. Due to the stickiness, we also have $\width(r_{\rm in})\leq  {\rm dist}_{\rm left}(v)$. Overall, we obtain 
\[
\width(R)\leq \varrho\cdot \height(R)
<\varrho \cdot \height(r_{\rm in}(v)) 
\leq 3\varrho \cdot \width(r_{\rm in}(v))
\leq 3\varrho \cdot {\rm dist}_{\rm left}(v).
\]
By the pigeonhole principle, $R$ crosses at most $3\varrho$ interior-disjoint cells between the left sides of $r_{\rm in}$ and $r_{\rm out}$. Similarly, $R$ crosses at most $3\varrho$ interior-disjoint cells between the right (respectively, top, bottom) sides of $r_{\rm in}$ and $r_{\rm out}$. As a result, $R$ crosses $O(\varrho)$ interior-disjoint cells in this case.
\end{proof}

\subsection{Extremal Points and their Properties}
\label{ssec:ext}
For each node $v$ of the BBD tree, we define a set ${\rm Ext}_v$ of a constant number of \emph{extremal points} in $P$. 

\begin{itemize}
    \item For an axis-aligned rectangle $r$, let ${\rm ext}(r)$ be a subset of $P\cap r$ that consists of a point with the minimum $x$-coordinate, maximum $x$-coordinate, minimum $y$-coordinate, and maximum $y$-coordinate (ties are broken arbitrarily).
    \item If $r_{\rm in}(v)=\emptyset$ (that is, $C_v  = r_{\rm out}(v)$), then let ${\rm Ext}_v={\rm ext}(r_{\rm out}(v))$; see~\Cref{fig_extreme}(a). 
    \item If $r_{\rm in}(v)\neq \emptyset$, then we subdivide  $C_v=r_{\rm out}(v)\setminus r_{\rm in}(v)$ along the lines spanned by the four sides of $r_{\rm in}(v)$ into $k$, $2\leq k\leq 9$, rectangular regions $C_v=\bigcup_{i=1}^k r_i$ and let ${\rm Ext}_v=\bigcup_{i=1}^k {\rm ext}(r_i)$; see~\Cref{fig_extreme}(b).
\end{itemize}
\begin{figure}[ht]
    \centering
     \vspace{-.4 cm}
    \includegraphics[page=3, scale=1]{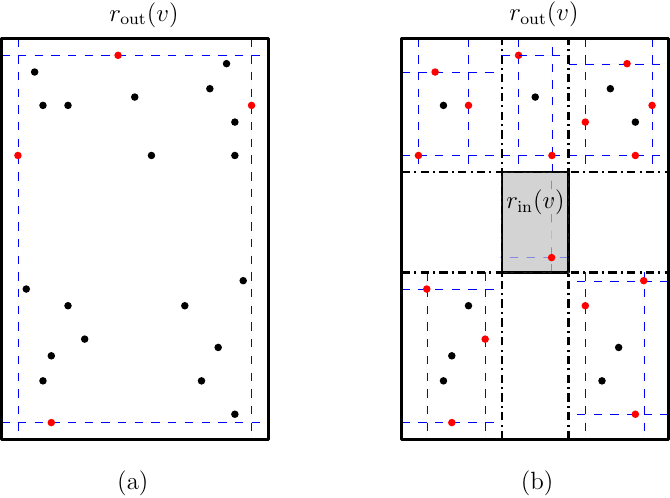}
    \caption{The extremal points of the cell $C_v$, i.e., ${\rm Ext}_v$ are colored red, when (a) $r_{\rm in}(v)=\emptyset$; (b) $r_{\rm in}(v)\neq\emptyset$.}
    \label{fig_extreme}
\end{figure}

\vspace{-\baselineskip}

\smallskip\noindent \emph{Properties of extremal points.}
We state a few properties of extremal points that will be used in the competitive analysis of our online algorithm in \Cref{sec:alg}.
Let $P$ be a finite set of points in a bounding box (a square), and $T$ be a BBD tree for $P$. 
    \begin{figure}[ht]
        \centering
        \includegraphics[page=7,scale=1.2]{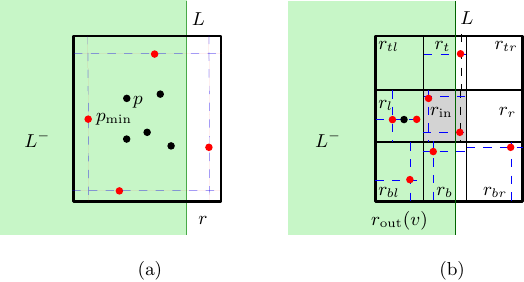}
        \caption{Illustration for (a)~\Cref{lem_line}; (b)~\Cref{lem:halfplane}.}
        \label{fig_properties}
        \vspace{-.2 cm}
    \end{figure}
    
\begin{lemma}\label{lem_line}
    Let $r$ be an axis-aligned rectangle, and $L^-$ be a half-plane bounded by an axis-parallel line $L$. If $P\cap r\cap L^-\neq \emptyset$, then ${\rm ext}(r)\cap L^-\neq \emptyset$.
\end{lemma}
\begin{proof}
Assume w.l.o.g. that $L$ is the vertical line $x = a$, for some $a\in \mathbb{R}$, and let $L^- = \{ (x, y) \in \mathbb{R}^2 : x \leq a \}$ be the left half-plane; see~\Cref{fig_properties}(a).  Since $P \cap r \cap L^- \neq \emptyset$, there exists a point $p\in P$ in the region $r\cap L^-$.
Note that all the points in $P \cap r \cap L^-$ lie within a region $r\cap L^-$ bounded on the left, top, and bottom by the corresponding edges of the rectangle $r$, and on the right by the line $L$.
Let  $p_{\min}$ be a point with the minimum $x$-coordinate in $P\cap r$.
Since $p \in P \cap r \cap L^-$, and $p_{\min}$ has $x$-coordinate no greater than $p$, we have $x(p_{\min}) \leq x(p) \leq a$, which implies that $p_{\min} \in L^-$.
Thus, we have $p_{\min} \in {\rm ext}(r) \cap L^-$, as required.
\end{proof}

\begin{lemma}\label{lem:halfplane}
    Let $C_v=r_{\rm out}(v)\setminus r_{\rm in}(v)$ be a cell (where $r_{\rm in}(v)$ may be empty), and $L^-$ be a half-plane bounded by an axis-parallel line $L$. 
    If $P\cap C_v\cap L^-\neq \emptyset$, then ${\rm Ext}_v\cap L^-\neq \emptyset$; see~\Cref{fig_properties}(b) for an illustration.
\end{lemma}
\begin{proof}
   If $r_{\rm in}(v) = \emptyset$, then cell $C_v = r_{\rm out}$, and so $C_v$ is a rectangle. In this case, ${\rm Ext}_v={\rm ext}(r_{\rm out})$ and~\Cref{lem_line} completes the proof.
    Now consider the case that $r_{\rm in}(v) \neq \emptyset$. Assume w.l.o.g.\ that $L$ is the vertical line $x = a$, for some $a\in \IR$, and let $L^- = \{ (x, y) \in \mathbb{R}^2 : x \leq a \}$ is the left half-plane; see~\Cref{fig_properties}(b). 
    Since $P \cap C_v \cap L^- \neq \emptyset$, there exists a point $p \in P\cap C_v \cap L^-$. 
    Recall that in the construction of ${\rm Ext}_v$, we subdivide $C_v=r_{\rm out}(v)\setminus r_{\rm in}(v)$, by extending the lines defined by the four sides of $r_{\rm in}(v)$, into up to eight subrectangles. 
    Assume that $p\in r_{x}$, where $r_x$ is one of the subrectangles of $C_v$. 
    Then $P\cap r_x\cap L^-\neq \emptyset$, and \Cref{lem_line} yields 
    ${\rm ext}(r_x)\cap L^-\neq \emptyset$. Since ${\rm ext}(r_x)\subset {\rm Ext}_v$ by definition, then ${\rm ext}(r_x)\cap L^-\neq \emptyset$, as claimed. 
\end{proof}

\begin{lemma}\label{lem:strip}
    Let $C_v=r_{\rm out}(v)\setminus r_{\rm in}(v)$ be a cell, and $S=L_1^-\cap L_2^+$ denotes the intersection of two half-planes bounded by two horizontal lines, or two vertical lines $L_1$ and $L_2$. 
    If $P\cap C_v\cap S\neq \emptyset$ and $S$ contains a corner of $r_{\rm in}(v)$, then ${\rm Ext}_v\cap S\neq \emptyset$; see~\Cref{fig_lem_strip}  for an illustration.
\end{lemma}
 \begin{figure}[ht]
        \centering
        \includegraphics[page=8,scale=1.2]{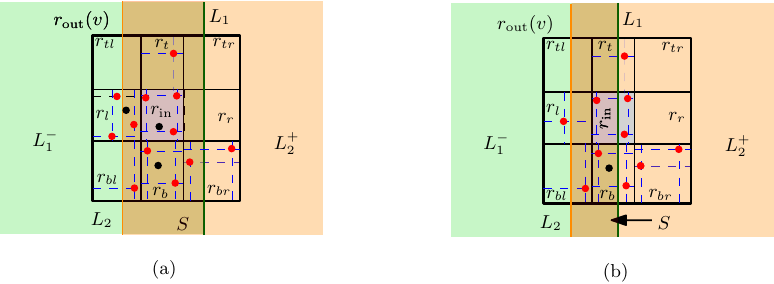}
        \caption{Illustration for~\Cref{lem:strip}.}
        \label{fig_lem_strip}
        \vspace{-.35 cm}
    \end{figure}

\begin{proof}
     Assume w.l.o.g.\ that  $L_1$ and $L_2$ are the vertical lines $x = a$ and $x=b$, respectively, for some $a,b\in \mathbb{R}$. Let $L_1^- = \{ (x, y) \in \mathbb{R}^2 : x \leq a \}$ be the left half-plane. If $L_2^+$ is also a left half-plane, then $L_1^-\cap L_2^+$ is a half-plane and \Cref{lem:halfplane} completes the proof. So we may assume that $L_2^+ = \{ (x, y) \in \mathbb{R}^2 : x \geq b \}$ is the right half-plane,
     and $b<a$ (since $S=L_1^-\cap L_2^+\neq \emptyset$); see~\Cref{fig_lem_strip}. 
     
     Assume that $S$ contains corner $c=(c_1,c_2)$ of $r_{\rm in}(v)$. Recall that $C_v$ is subdivided along the lines spanned by the sides of $r_{\rm in}(v)$. One of the subdivision lines is $x=c_1$, where $a<c_1<b$ (since $c\in S$). Consequently, every subrectangle of $C_v$ intersects at most one of $L_1$ and $L_2$. 

     Since $P\cap C_v\cap S\neq \emptyset$, then there is a point $p\in P\cap C_v\cap S$. Assume that $p\in r_x$, where $r_x$ is one of the subrectangles of $C_v$.
     Since ${\rm ext}(r_x) \cap S \subseteq {\rm Ext}_v \cap S$ and by~\Cref{lem_line}, we have ${\rm ext}(r_x) \cap S\neq \emptyset$, thus we have ${\rm Ext}_v \cap S\neq \emptyset$. Hence, the lemma follows. 
\end{proof}

\begin{lemma}\label{lem:corner}
    Let $C_v=r_{\rm out}(v)\setminus r_{\rm in}(v)$ be a cell, let $S=L_1^-\cap L_2^+$ be the intersection of two half-planes bounded by two perpendicular axis-parallel lines $L_1$ and $L_2$. 
    If $P\cap C_v\cap S\neq \emptyset$ and $L_1\cap L_2\in r_{\rm in}(v)$, then ${\rm Ext}_v\cap S\neq \emptyset$; see~\Cref{fig_lem_corner} for an illustration.
\end{lemma}
\begin{figure}[htbp]
        \centering
        \includegraphics[page=9,scale=1.2]{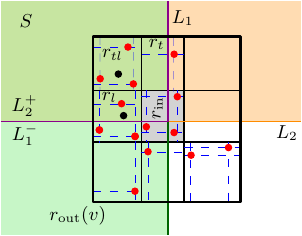}
        \caption{Illustration for~\Cref{lem:corner}.}
        \label{fig_lem_corner}
        \vspace{-.2 cm}
    \end{figure}
\begin{proof}
    Assume w.l.o.g.\ that $L_1$ is the vertical line $x = a$, and $L_2$ is the horizontal line $y = b$ for some $a,b\in \mathbb{R}$. Let $L_1^- = \{ (x, y) \in \mathbb{R}^2 : x \leq a \}$ be the left half-plane and $L_2^+ = \{ (x, y) \in \mathbb{R}^2 : y \geq b \}$ be the top half-plane; see~\Cref{fig_lem_corner}. 
    Since $L_1 \cap L_2 = (a, b) \in r_{\rm in}(v)$, 
    then $L_1$ intersects the subrectangles of $C_v$ above and below $r_{\rm in}$(v); and 
    $L_2$ intersects the subrectangles of $C_v$ to the left and right of $r_{\rm in}(v)$.
    In particular, none of the subrectangles in $C_v$ intersects both $L_1$ and $L_2$. 

    Since $P\cap C_v\cap S\neq \emptyset$, then there is a point $p\in P\cap C_v\cap S$. Assume that $p\in r_x$, where $r_x$ is one of the subrectangles of $C_v$. Since $r_x$ intersects at most one of $L_1$ and $L_2$, 
    then $r_x\cap S$ equals $r_x\cap L_1^-$ or $r_x\cap L_2^+$. 
    In both cases, by~\Cref{lem_line}, we have ${\rm ext}(r_x) \cap S\neq \emptyset$.
    Now ${\rm ext}(r_x) \subseteq {\rm Ext}_v$ implies that 
    ${\rm Ext}_v \cap S\neq \emptyset$, as required.
\end{proof}

\section{Online Algorithm and Competitive Analysis}
\label{sec:alg}

In this section, we prove \Cref{thm:main}: We present our online hitting set algorithm for axis-aligned rectangles of aspect ratio $\varrho$, for constant $\varrho\geq 1$, in \Cref{sec_online_alg}, and analyze its competitive ratio in \Cref{ssec:analysis}

\subsection{Online Algorithm}
\label{sec_online_alg}

We now describe our online algorithm. Given a set $P$ of $n$ points in the plane, we compute a BBD tree $T$ for $P$, and the set $H_v$ of extremal points ${\rm Ext}_v\subset P$ for all nodes $v\in V(T)$. As the adversary presents axis-aligned rectangles $\{S_1,\ldots, S_m\}$ of aspect ratio at most $\varrho$, we  maintain the following data structures: 

A hitting set $H_i\subset P$ for the first $i$ rectangles $\{S_j: j\leq i\}$, which is initialized to $H_0=\emptyset$. A set $A_i\subset V(T)$ of \emph{active} nodes of $T$ with the following property: If a node $v\in V(T)$ is active, then all ancestors of $v$ are also active. Initially, all nodes are \emph{inactive} (i.e., $A_0=\emptyset$); inactive nodes can become active, and any active node remains active for the remainder of the algorithm. Furthermore, we maintain the property that $\bigcup_{v\in A_i} {\rm Ext}_v\subset H_i$, i.e., $H_i$ contains the extremal points of all active nodes (however, $H_i$ may contain additional points, as well). 

When a rectangle $S_i$ arrives, initialize $H_i:=H_{i-1}$ and $A_i:=A_{i-1}$. If $S_i\cap H_i\neq \emptyset$, then no further changes are needed in our data structures. Suppose that $S_i\cap H_i = \emptyset$. 
\begin{enumerate}
    \item Consider the four corners of $S_i$. 
     \begin{enumerate}
        \item  For each corner $a\in C_{\rm root}$, find the highest inactive node $v\in V(T)$ such that $a\in C_{v}$: activate $v$ and its sibling $v'$ (set $A_i:=A_i\cup \{v,v'\}$) and add their extremal points to $H_i$ (set $H_i:=H_i\cup {\rm Ext}_{v}\cup {\rm Ext}_{v'}$). 
        \item For each corner $a\notin C_{\rm root}$, if the root of $T$ is not already active, then activate the root and add its extremal points to $H_i$.
            \end{enumerate} 
    \item For every node $u\in V(T)$ where $S_i$ crosses $C_u$, 
    \begin{enumerate}
        \item If $u$ is already active but its children are not, then activate both children of $u$, and add their extremal points to $H_i$.
        \item If $u$ is inactive, then find the highest inactive ancestor $v$ of $u$ (possibly $v=u$): activate $v$ and its sibling $v'$ and add their extremal points to $H_i$.
    \end{enumerate} 
    \item If $S_i\cap H_i= \emptyset$ still holds, then add an arbitrary point in $S_i\cap P$ to $H_i$. 
\end{enumerate}


\subsection{Competitive Analysis}
\label{ssec:analysis}

The algorithm guarantees that $H_i$ is a hitting set for the first $i$ objects $\{S_1,\ldots , S_i\}$. It is also clear that for each new object $S_i$, we add $O(\varrho)$ new points to $H_i$: Since $S_i$ has four corners and crosses $O(\varrho)$ cells of the BBD tree by \Cref{lem:rho}. 

Let $\opt\subset P$ be an offline optimum, i.e., a minimum hitting set for $\mathcal{C}=\{S_1,\ldots, S_m\}$. For each point $p\in \opt$, let $\mathcal{C}_p=\{S_j\in \mathcal{C}: p\in S_j\}$ be the set of objects hit by $p$. 
It is sufficient to show that for every $p$, the algorithm adds new points to the hitting set in $O(\log n)$ steps to hit objects in $\mathcal{C}_p$. Since $O(\varrho)$ points are added to the hitting set in each step, then the algorithm adds $O(\varrho\log n)$ points in response to the objects in $\mathcal{C}_p$, hence $O(|\opt|\cdot \varrho\log n)$ points in response to all objects in $\mathcal{C}$.  

\begin{lemma}
    If $S_i\in \mathcal{C}_p$ and $S_i\cap H_{i-1}= \emptyset$, then in step $i$, the algorithm activates a cell of the BBD tree containing $p$.
\end{lemma}
\begin{proof}
    The algorithm activates the root in step $i=1$. In the remainder of the proof, we may assume that $i>1$ and the root is already active. 
    Before the arrival of $S_i$, we have $p\notin H_{i-1}$, so the leaf node of the BBD tree that contains $p$ is inactive. Let $v$ be the lowest active node in the BBD tree such that $C_v$ contains $p$. Let $u$ and $w$ be the two children of $v$ such that $p\in C_u$ (hence $p\notin C_w$). 
    
    We need to show that the algorithm activates $u$ in step~$i$. Suppose, for the sake of contradiction, that $u$ is inactive at the beginning of step~$i$. Then its sibling $w$ is also inactive at that time (since our algorithm always activates two siblings). 
    
   The algorithm activates the highest inactive nodes that contain any of the four corners of $S_i$ (and their siblings), as well as the highest inactive nodes corresponding to every cell crossed by $S_i$ (and their siblings). Therefore, we may assume that neither $C_u$ nor $C_w$ contains any corner of $S_i$, and neither of them is crossed by $S_i$. Since $C_v=C_u\cup C_w$, then $C_v$ does not contain any corner of $S_i$, either. If  $S_i$ crosses $C_v$, then at the beginning of step~$i$, node $v$ is active but its children $u$ and $w$ are inactive, and so the algorithm would activate both $u$ and $w$ in step~$i$. For this reason, we may also assume that $S_i$ does not cross $C_v$.

  We examine all possible positions of $C_v$ and $C_u$ relative to $S_i$. Recall that $C_v=r_{\rm out}(v)\setminus r_{\rm in}(v)$, where $r_{\rm in}(v)$ may be empty.

  \begin{figure}[htbp]
    \centering
    \includegraphics[page=4,scale=1]{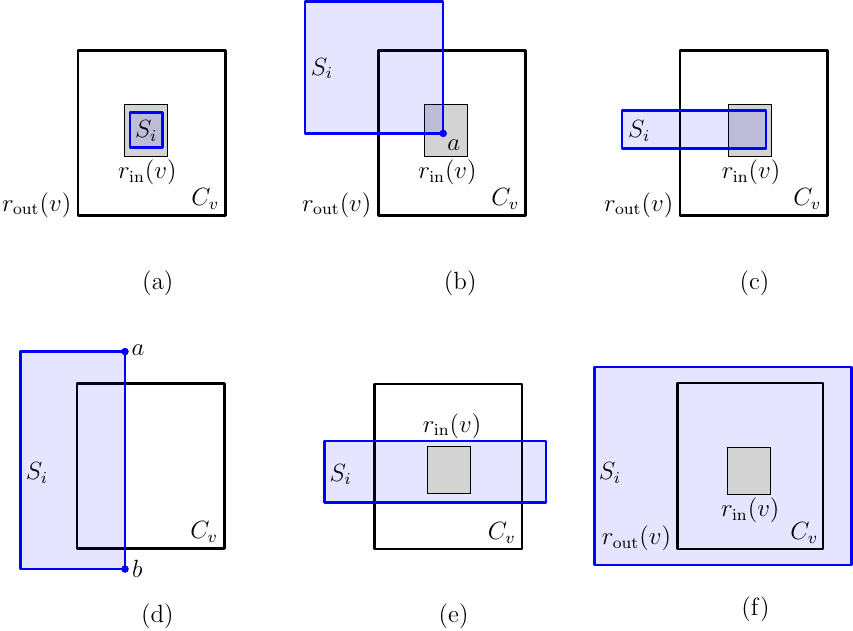}
    \caption{Illustration of (a) Case 1; (b) Case 2a; (c) Case 2b; (d) Case 3a; (e) Case 3b; and (f) Case 4.}
    \label{fig_cases}
    \vspace{-.35 cm}
\end{figure}

\smallskip\noindent
\textbf{Case~1: $S_i\subset r_{\rm out}(v)$; see~\Cref{fig_cases}(a).} 
In this case, all four corners of $S_i$ are in $r_{\rm in}(v)$. If all four corners of $S_i$ are in $r_{\rm in}(v)$, then $p\in S_i\subset r_{\rm in}(v)$, which contradicts the assumption that $p\in C_v$. Therefore, a corner of $S_i$ lies in $r_{\rm out}(v)\setminus r_{\rm in}(v)=C_v$. For this corner of $S_i$, the highest inactive node is $u$ or $w$. Consequently, the algorithm activates both siblings $u$ and $w$. In particular, $u$ is activated. 

\smallskip\noindent
\textbf{Case~2: $S_i\not\subset r_{\rm out}(v)$ but $r_{\rm out}(v)$ contains some corner of $S_i$; see~\Cref{fig_cases}(b-c).} Since $C_v$ does not contain any corners of $S_i$, then all corners of $S_i$ in $r_{\rm out}(v)$ are in $r_{\rm in}(v)$. Since $S_i\not\subset r_{\rm out}(v)$, then $r_{\rm in}(v)$ contains either one or two corners of $S_i$. We examine each case separately: 

\smallskip\noindent
\textbf{Case~2a: $r_{\rm in}(v)$ contains precisely one corner of $S_i$}. Assume w.l.o.g.\ that $r_{\rm in}(v)$ contains the lower-right corner of $S_i$ (as in~\Cref{fig_cases}(b)). Denote by $a$ the lower-right corner of $S_i$. Then $S_i$ also contains precisely one corner of $r_{\rm in}(v)$, namely its upper-left corner. Since the remaining three corners of $S_i$ are outside of both $r_{\rm in}(v)$ and $C_v$, they are outside of $r_{\rm out}(v)$. Consequently, $S_i$ also contains precisely one corner of $r_{\rm out}(v)$, namely the upper-left corner of $r_{\rm out}(v)$.
Since $p\in S_i\cap C_v$, then \Cref{lem:corner} yields ${\rm Ext}_v\cap S_i\neq \emptyset$.
However, $C_v$ was activated in a previous step. This implies that ${\rm Ext}_v\subset H_{i-1}$, hence $S_i\cap H_{i-1}\neq \emptyset$: a contradiction.

\smallskip\noindent
\textbf{Case 2b: $r_{\rm in}(v)$ contains exactly two corners of $S_i$.} Recall that an axis-aligned rectangle $S_i$ crosses a cell $C_v$, for $v\in V(T)$ if $C_v\cap S_i\neq \emptyset$ but $C_v$ does not contain any vertex of $S_i$ and $S_i$ does not contain any vertex of $C_v$. Consequently, $S_i$ crosses $C_v$: a contradiction.

\smallskip\noindent
\textbf{Case~3: $r_{\rm out}(v)$ does not contain any corner of $S_i$, but it intersects some edges of $S_i$; see Figure~\ref{fig_cases}(d-e).} 
Let $ab$ be an edge of $S_i$ that intersects $C_v$, where $a$ and $b$ are corners of $S_i$. Since $r_{\rm out}(v)$ contains neither $a$ nor $b$, then both $a$ and $b$ are outside of $r_{\rm out}(v)$, and so the two edges of $S_i$ orthogonal to $ab$ are also outside of $r_{\rm out}(v)$. Consequently, $r_{\rm out}(v)$ intersects one edge of $S_i$ or two parallel edges of $S_i$. 

\smallskip\noindent
\textbf{Case~3a: $r_{\rm out}(v)$ intersects precisely one edge of $S_i$.} Assume w.l.o.g.\ that $r_{\rm out}(v)$ intersects the right edge $ab$ of $S_i$; see~\Cref{fig_cases}(d). Let $L$ be the vertical line spanned by $ab$, and let $L^-$ be the left half-plane determined by $L$. 
Note that $p\in S_i\cap C_v$ implies $p\in L^-$, so $P\cap C_v\cap L^-\neq \emptyset$.
\Cref{lem:halfplane} yields ${\rm Ext}_v\cap S_i\neq \emptyset$.
However, $C_v$ was activated in a previous step. This implies that ${\rm Ext}_v\subset H_{i-1}$, hence $S_i\cap H_{i-1}\neq \emptyset$: a contradiction. 


\smallskip\noindent
\textbf{Case~3b: $r_{\rm out}(v)$ intersects two parallel edges of $S_i$.}
Assume w.l.o.g.\ that $r_{\rm out}(v)$ intersects the left and right edges of $S_i$. 
If $S_i$ does not contain any corners of $r_{\rm in}(v)$, then $S_i$ does not contain any vertex of $C_v$, and so $S_i$ crosses $C_v$: a contradiction.

So we may assume that $S_i$ contains some corners of $r_{\rm in}(v)$.
\Cref{lem:strip} yields ${\rm Ext}_v\cap S_i\neq \emptyset$. However, $C_v$ was activated in a previous step. This implies that ${\rm Ext}_v\subset H_{i-1}$, hence $S_i\cap H_{i-1}\neq \emptyset$: a contradiction. 


\smallskip\noindent
\textbf{Case~4: $r_{\rm out}(v)$ lies in the interior of $S_i$; see Figure~\ref{fig_cases}(f).}  Since $p\in P\cap C_v$, then $P\cap C_v\neq \emptyset$, and so ${\rm Ext}_v\neq \emptyset$. Since $C_v$ is active, then ${\rm Ext}_v\subset H_{i-1}$, and $C_v\subset r_{\rm out}(v)\subset S_i$ implies that ${\rm Ext}_v\subset S_i$.
Consequently, $H_{i-1}\cap S_i\neq \emptyset$:  a contradiction.
\end{proof}

\section{Generalizations to Positive Homothets of Polygons}\label{sec:generalizations}
In \Cref{sec:alg}, we presented an $O(\log n)$-competitive algorithm for the online hitting set problem with $n$ points in the plane and axis-aligned squares (of aspect ratio  $\varrho= 1$). Axis-aligned squares are homothets of a unit square. Since a suitable linear transformation takes any parallelogram into a unit square, our result immediately extends to positive homothets of a parallelogram. 

\begin{corollary}\label{cor:parallelogram}
For every parallelogram $M$ in the plane, there is an $O(\log n)$-competitive algorithm for the online hitting set problem for any set of $n$ points in the plane and a sequence of positive homothets of $M$. 
\end{corollary}

We can further generalize \Cref{thm:main} to positive homothets of a polygon made of finitely many parallelograms (\Cref{thm:generalization}). 

\begin{lemma}\label{lem:union}
  Every polygon with $k\geq 3$ vertices is the union of at most $5k-12$ parallelograms.
\end{lemma}
\begin{proof}
    Every simple polygon $M$ with $k\geq 3$ vertices admits a triangulation with $k-2$ triangles. In general, a polygon $M$ with $k\geq 3$ vertices and $h\geq 0$ holes admits a triangulation with $k+2h-2$ triangles~\cite[Lemma~5.2]{ORourke87}. Since $h\leq k/3-1$, then $M$ is always the union of at most $5k/3-4$ triangles. 
    Every triangle $T$ is decomposed, by its three medians, into four congruent subtriangles: One containing the center of $T$, and three incident to each of the corners of $T$. The union of the central subtriangle and a corner subtriangle is a parallelogram; see \Cref{fig:union}. Thus $T$ is the union of three parallelograms, and consequently, $M$ is the union of at most $3(5k/3-4)=5k-12$ parallelograms.
\end{proof}

\begin{figure}[htbp]
    \centering
    \vspace{-\baselineskip}
    \includegraphics[scale=1]{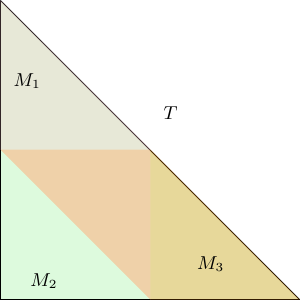}
    \caption{A triangle is the union of three parallelograms.}
    \label{fig:union}
    \vspace{-.2 cm}
\end{figure} 

\begin{lemma}\label{lem:combination}
    Let $T$ be a polygon that can be written as a union of $k$ parallelograms. Then there is an $O(k^2\log n)$-competitive deterministic algorithm for the online hitting set problem for any set $P$ of $n$ points in the plane, and a sequence of positive homothets of $T$.  
\end{lemma}
\begin{proof}
    Assume that $T$ is the union of $k$ parallelograms, i.e., $T=\bigcup_{j=1}^k M_j$, where each $M_j$ is a parallelogram.  By \Cref{cor:parallelogram}, there is an $O(\log n)$-competitive deterministic algorithm $\alg_j$ for the online hitting set problem for the same point set $P$ and a sequence of positive homothets of the parallelogram $M_j$, for every $j\in \{1,\ldots, k\}$.  

\smallskip\noindent\emph{Online algorithm.} 
   We now describe a deterministic online algorithm for the point set $P$ and a sequence $T_1, T_2,\ldots$ of positive homothets of $T$. For every $i\in \mathbb{N}$, we have 
   $T_i=a_iT+b_i$ for some scaling factor $a_i>0$ and translation vector $b_i\in \mathbb{R}^2$. Since $T=\bigcup_{j=1}^k M_j$, then 
   $T_i=\bigcup_{j=1}^k M_{i,j}$, where $M_{i,j}=a_i M_j+b_i$ is a positive homothet of the parallelogram $M_j$. 
  
   We maintain a hitting set $H_i\subset P$ for the first $i$ homothets $\{T_\ell: \ell\leq i\}$, which is initialized to $H_0=\emptyset$. We initialize the algorithm $\alg_j$ for $j=1,\ldots, k$, that each maintain a hitting set $H_{i,j}\subset P$ for some \emph{subset} of the first $i$ parallelograms $\{M_{\ell,j}: \ell\leq i\}$. We maintain that $H_i=\bigcup_{j=1}^k H_{i,j}$. 

   When a homothet $T_i=a_iT+b_i$ arrives. We initialize $H_i:=H_{i-1}$, and $H_{i,j}:=H_{i-1,j}$ for all $j=1,\ldots ,k$. If $T_i\cap H_i\neq \emptyset$, then no further changes are needed. Otherwise, we compute the parallelograms $M_{i,j}=a_iM_j+b_i$ for $j=1,\ldots , k$. For each $j=1,\ldots , k$, if $P\cap M_{i,j}\neq \emptyset$, then we feed $M_{i,j}$ to the algorithm $\alg_j$, which in turn adds new points to $H_{i,j}$.
   Finally, we update $H_i$ by setting $H_i:=\bigcup_{j=1}^k H_{i,j}$. 
   This completes the description of the algorithm. 

    \smallskip\noindent \emph{Competitive analysis.} The algorithm guarantees that $H_i$ is a hitting set for the first $i$ objects $\{T_1,\ldots , T_i\}$. It is also clear that for each new object $T_i$, we add $O(k)$ new points to $H_i$: Since each algorithm $\alg_j$ adds $O(1)$ points to $H_{i,j}$.  

    Let $\opt\subset P$ be an offline optimum, i.e., a minimum hitting set for $\mathcal{C}=\{T_1,\ldots , T_m\}$. For each point $p\in \opt$, let $\mathcal{C}_p=\{T_i\in \mathcal{C}: p\in T_i\}$ be the set of objects hit by $p$. It is sufficient to show that for every $p$, the algorithm adds $O(k^2\log n)$ points to $H_i$ to hit objects in $\mathcal{C}_p$.

    Let $\mathcal{C}_p'$ be a set of objects $T_i\in \mathcal{C}_p$ such that our algorithm adds new points to the hitting set in step~$i$. Since our algorithm adds $O(k)$ points to the hitting set in each step, it is enough to show that $|\mathcal{C}_p'|\leq O(k\log n)$.     
    We further partition $\mathcal{C}_p'$ based on which parallelograms are hit by point $p$. For $j=1,\ldots , k$, let $\mathcal{C}_{p,j}' =\{T_i\in \mathcal{C}_p' : p\in M_{i,j}, \mbox{ \rm and } p\notin M_{i,j'} \mbox{ \rm for } j'<j\}$; and $\mathcal{M}_{p,j}' =\{M_{i,j}: T_i\in  \mathcal{C}_{p,j}'\}$. 
    By definition, $\mathcal{C}_p=\bigcup_{j=1}^k \mathcal{C}_{p,j}'$. 
    Note that $\mathcal{M}_{p,i}'$ is a set of homothets of the parallelogram $M_j$ that contain $p\in P$. By assumption, algorithm $\alg_j$ is $O(\log n)$-competitive, so it adds $O(\log n)$ points to the hitting set in response to parallelograms in $\mathcal{M}_{p,j}'$. Recall that we feed a parallelogram $M_{i,j}\in \mathcal{M}_{p,j}'$ to $\alg_j$ only if object $T_i$ has not been hit by $H_{i-1}$, 
    and so $\alg_j$ adds at least one new point to the hitting set in step $i$.
    Consequently, we have $|\mathcal{M}_{p,j}'|\leq O(\log n)$. Overall, we obtain  $|\mathcal{C}_p'|=\sum_{j=1}^k |\mathcal{C}_{p,j}'| = \sum_{j=1}^k |\mathcal{M}_{p,j}'|\leq O(k\log n)$, as required. Hence, the theorem follows.
    \end{proof}
   
The combination of \Cref{lem:union,lem:combination} immediately implies \Cref{thm:generalization}.


\section{Conclusions}
\label{sec:conclusion}

Our main result (\Cref{thm:main}) is an $O(\log n)$-competitive algorithm for the online hitting set problem for $n$ points in the plane, and a sequence of axis-aligned squares (or axis-aligned rectangles of bounded aspect ratio). This is the first online hitting set algorithm that is $O(\log n)$-competitive for geometric objects of arbitrary sizes in the plane. Our result further generalizes to positive homothets of any simple $k$-gon for $k\geq 3$, and achieves $O(k^2\log n)$-competitive algorithm. Even though a disk can be approximated by a regular $k$-gon as $k\to \infty$, our generalized result does not imply any competitive algorithm for disks. It remains an open problem whether an $O(\log n)$-competitive online hitting set algorithm for $n$ points and disks (of arbitrary radii), or positive homothets of a polygon $P$ (with arbitrarily many vertices) exists.

Does our result generalize to higher dimensions? Is there $O(\log n)$-competitive online hitting set algorithm for axis-aligned cubes of arbitrary sizes in $\mathbb{R}^d$ for constant dimension $d$? Is there one for $d=3$?
While BBD trees generalize to $\mathbb{R}^d$, for any constant dimension $d\in \mathbb{N}$, our online algorithm does not generalize to (hyper-)cubes in $d$-space. The reason is that the extremal points (cf.~\Cref{ssec:ext}) do not necessarily capture the intersection $C_v\cap R$ of a cell $C_v$ of the BBD tree and an axis-aligned box. In the plane, we have shown that $C_v$ contains a vertex of $R$, or $C_v\cap R$ is a crossing intersection, or $C_v\cap R$ behaves as an axis-aligned halfplane with respect to the sub-rectangles of $C_v$ (\Cref{ssec:ext}). However, in $\mathbb{R}^d$, $d\geq 3$, it is possible that $R$ contains exactly one edge of cell $C_v=r_{\rm out}(v)$ and yet $C_v\cap R$ does not contain any of the extremal points of $C_v$. It remains open whether $O(\log n)$-competitive algorithms are possible for geometric objects of arbitrary sizes in dimensions $d\geq 3$.

\bibliographystyle{plainurl}
\bibliography{square}
\end{document}